\documentclass{elsarticle}
\usepackage{graphicx,float}
\usepackage{amsmath,geometry,amssymb,amsthm}
\usepackage{caption}
\usepackage{subcaption}
\usepackage{tikz}
\newtheorem{prop}{Proposition}
\newtheorem{thm}{Theorem}

\newtheorem{cor}{Corollary}

\newtheorem{definition}{Definition}

\begin{document}


\begin{frontmatter}

\title{Algorithmic Complexity of Secure Connected Domination in Graphs}
\author{Jakkepalli Pavan Kumar}
\ead{jp.nitw@gmail.com}
\author{P. Venkata Subba Reddy}
\address{Department of Computer Science and Engineering\\
	National Institute of Technology Warangal, Warangal, Telangana, India}
\ead{pvsr@nitw.ac.in}
\author{S. Arumugam}
\address{Director, \textit{n}-CARDMATH, \\
	Kalasalingam Academy of Research and Education\\
	Anand Nagar, Krishnankoil, Tamilnadu, India.}
\ead{s.arumugam.klu@gmail.com}
\begin{abstract}
Let $G = (V,E)$ be a simple, undirected and connected graph. A connected (total) dominating set $S \subseteq V$ is a secure connected (total) dominating set of $G$, if for each $ u \in V \setminus S$, there exists $v \in S$ such that $uv \in E$ and $(S \setminus \lbrace v \rbrace) \cup \lbrace u \rbrace $ is a connected (total) dominating set of $G$. The minimum cardinality of a secure connected (total) dominating set of $G$ denoted by $ \gamma_{sc} (G) (\gamma_{st}(G))$, is called the secure connected (total) domination number of $G$.  In this paper, we show that the decision problems corresponding to secure connected domination number and secure total domination number are NP-complete even when restricted to split graphs or bipartite graphs. The NP-complete reductions also show that these problems are w[2]-hard. We also prove that the secure connected domination problem is linear time solvable in block graphs and threshold graphs.

\end{abstract}


\end{frontmatter}

\section{Introduction}\label{intro}
Let $G(V,E)$ be a simple, undirected and connected graph. For graph theoretic terminology we refer to \cite{west}. For a vertex $v \in V$, the \textit{open neighborhood} of $v$ in $G$ is $N_G(v)$= \{$u \in V: uv \in E$\}, the \textit{closed neighborhood} of $v$ is defined as $N_G[v]=N_G(v) \cup \{v\}$. If $S \subseteq V$, then the open neighborhood of $S$ is the set $N_G(S) = \cup_{v \in S} N_G(v)$. The closed neighborhood of $S$ is $N_G[S] = S \cup N_G(S)$. Let $ S \subseteq V$. Then a vertex $ w \in V $ is called a private neighbor of $ v $ with respect to $ S $ if $ N[w] \cap S = \{v\}. $ If further $ w \in V \setminus S$, then $ w $ is called an \textit{external private neighbor (epn)} of $v$.  \par 
A subset $S$ of $V$ is a \textit{dominating set} (DS) in $G$ if for every $u \in V \setminus S$, there exists $v \in S$ such that $uv \in E$. The \textit{domination number} of $G$ is the minimum cardinality of a DS in $G$ and is denoted by $\gamma(G)$. A set $S \subseteq V$ is said to be a \textit{secure dominating set} (SDS) in $G$ if for every $u \in V \setminus S$ there exists $v \in S$ such that $uv \in E$ and $(  S \setminus \{ v \})  \cup \{ u \} $ is a dominating set of $G$.  We say that $ v $ $ S$-\textit{defends} $u $ or $ u $ is defended by $ v $. The minimum cardinality of a SDS in $G$ is called the \textit{secure domination number} of $G$ and is denoted by $\gamma_s(G)$. A dominating set $S$ is said to be a \textit{connected dominating set} (CDS), if the induced subgraph $G[S]$ is connected. A CDS $S$ is said to be a \textit{secure connected dominating set} (SCDS) in $G$ if for each $u \in V \setminus S$, there exists $v \in S$ such that  $uv \in E$ and $(S \setminus \{ v \})  \cup \{ u \} $ is a CDS in $G$. The minimum cardinality of a SCDS in $G$ is called the \textit{secure connected domination number} of $G$ and is denoted by $\gamma_{sc}(G)$. A dominating set $S$ is said to be a \textit{total dominating set} (TDS), if the induced subgraph $G[S]$ has no isolated vertices. A TDS $S$ is said to be a \textit{secure total dominating set} (STDS) of $G$, if for each $u \in V \setminus S$, there exists $v \in S$ such that  $uv \in E$ and $( S \setminus \{ v \})  \cup \{ u \}$ is a TDS in $G$. The minimum cardinality of a STDS in $G$ is called the \textit{secure total domination number} of $G$ and is denoted by $\gamma_{st}(G)$.
We need the following theorems.
\begin{thm}(\cite{scd})\label{com}
	Let $ G $ be a connected graph of order $ n $. Then $ \gamma_{sc}(G)=1 $ if and	only if $ G = K_n $.
\end{thm}

\begin{thm}(\cite{scd})\label{pendant}
	Let $ G $ be a connected graph of order $ n \ge 3 $. Let $ L(G) $ and $ S(G) $ be the set of pendant and support vertices of $ G $ respectively. Let $ X $ be a secure connected dominating set of $ G $. Then
	(i) $ L(G) \subseteq X $ and $ S(G) \subseteq X $\\
	(ii) No vertex in $ L(G) \cup S(G) $ is an X-\textit{defender}.
\end{thm}
\begin{prop}(\cite{sc2})\label{p1}
	Let $S$ be a CDS in $G$. Then $S$ is a SCDS in $G$ if and only if the following conditions are satisfied.
	\begin{enumerate}
		\item [(i)] $epn(v, S) = \emptyset$ for all $v \in S$.
		\item[(ii)] For every $u \in V \setminus S$, there exists $v \in S \cap N_G(u)$ such that $V(C) \cap N_G(u) \ne \emptyset$ for every component $C$ of $G[S \setminus \{v\}] $. 
	\end{enumerate}	
\end{prop} 
\begin{prop}(\cite{scd})\label{p2}
	Let $G$ be a non-complete connected graph and let $S$ be a secure connected dominating set in $G$. Then the set $S \setminus \{v\}$ is a dominating set for
	every $v \in S$. In particular, $1 +\gamma(G) \le \gamma_{sc}(G)$.
\end{prop}

\section{Main Results}
\noindent We first determine the value of $ \gamma_{sc}(G) $ for two families of graphs.

				\begin{thm}
				Let $ W_n = v_1+C_n$ be the wheel of order $ n+1 $ where $ n \ge 3 $. Let $ G $ be the graph obtained from $ W_{n+1} $ by subdividing all the edges of $ C_n $. Then $ \gamma_{sc}(G)=n+1.$
				\end{thm}
				\begin{proof}
					Let $V(G)=\{v_1,v_2,\ldots,v_{2n+1}\}$, $ d(v_{2n+1})=n $, $ d(v_i)=\begin{cases}
					2 & 
					\text{if } i \text{ is even}  \\
					3 & \text{otherwise}  
					\end{cases}
					 $ and \\$ N(v_i)=\{v_{i-1}, v_{i+1}\} $ if $ i $ is even. Then $ S=\{v_i : i$ is odd$\} $ is a SCDS of $ G.$ 
					Hence $\gamma_{sc}( G ) \le n+1$. \par
					Now let $ D $ be any $ \gamma_{sc}$-set of $ G $. If $ v_{2n+1} \notin D$ or if $ v_{2n+1} \in D$ and defends a vertex $ v_i $, then we get a connected dominating set $ D_1 $ of $ G $ such that $ \vert D_1 \vert = \vert D \vert$ and $ v_{2n+1} \notin D_1$. Hence $ \vert D\vert= \vert D_1 \vert \ge 2n-2$, which is a contradiction. Thus $ v_{2n+1} \in D $ and $ v_{2n+1}$ does not defend any other vertex. Now let $ v_i \in D$ for some $ i$ where $ i $ is even. Since $ G[D] $ is connected, one of $ v_{i-1} $ or $ v_{i+1} $ is in $ D.$ Also if $ v_i \notin D $ for all even $ i,$ then $ v_{i} \in D$ for all odd $ i.$ Hence $ \gamma_{sc}(G)= \vert D \vert \ge n+1.$
				\end{proof}	
				\begin{thm}
					For the Book graph $ B_n=K_{1,n} \Box K_2 $, we have $\gamma_{sc}( B_{n} )$ = $n+2$.
						\end{thm}
						\begin{proof}
							Let $ S_1 $ and $ S_2 $ be the two copies of $ K_{1,n} $ in $ B_n.$ Let $ V(S_1)=\{v_1,v_2,\ldots,v_{n+1}\} $ and $ V(S_2)=\{w_1,w_2,\ldots,w_{n+1}\} $. Let $ v_1 $ and $ w_1 $ be the central vertices of $ S_1, S_2 $ respectively. Let $ v_iw_i \in E(B_n) $. Clearly $ V(S_1) \cup \{w_1\} $ is an SCDS of $ B_n.$ Hence $ \gamma_{sc}(B_n) \le n+2.$\par
							Now let $ D $ be any $ \gamma_{sc}$-set of $ B_n.$ Since $ D $ is connected, either $ v_1 $ or $ w_1 $ is in $ D.$ If $ w_1 \in D$ and $ v_1 \notin D $, then $ \{w_2,w_3,\ldots,w_{n+1},v_2,v_3,\ldots,v_{n+1}\} \subseteq D.$ Thus, $ \vert D \vert \ge 2n+1 $ which is a contradiction. Hence $ v_1,w_1 \in D.$ Now if both $ w_i $ and $ v_i $ are not in $ D $ for some $ i\ge 2, $ then $G[(D\setminus\{w_1\}) \cup \{w_i\}] $ and $ G(D\setminus\{v_1\}) \cup \{v_i\}] $ are disconnected. Hence $ \vert D \cap \{w_i,v_i\} \vert \ge 1 $ for any $ i \ge 2 $ and $ \gamma_{sc}( B_{n} )=\vert D \vert \ge n+2.$ Thus, $\gamma_{sc}( B_{n} )$ = $n+2$.
						\end{proof}

\begin{thm}
	Let $ G=P_n \Box P_2 $ where $ n \ge 3.$ Then $ \gamma_{sc}(G)=n+\lceil\frac{n}{3} \rceil.$
\end{thm}
\begin{proof}
	Let $ P=(v_1,v_2,\ldots,v_n) $ and $ Q=(w_1,w_2,\ldots,w_n) $ be two copies of $ P_n $ in $ G$ such that $ v_iw_i \in E(G).$ Let $ V_1=\{v_1,v_2,\ldots,v_n\}$ and $ V_2=\{w_1,w_2,\ldots,w_n\}.$ Then $ S=V_1 \cup \{w_i : i \equiv 2 (mod 3)\} $ is a SCDS of $ G $. Hence $ \gamma_{sc}(G) \le n+\lceil\frac{n}{3} \rceil.$ \par
	Let $ D $ be any $ \gamma_{sc}$-set of $ G.$ If $ v_i,w_i \notin D $ for some $ i,$ where $ 2 \le i \le n-1 $, then $ G[D] $ is disconnected, which is a contradiction. Hence at least one of $ v_i, w_i $ is in $ D$, where $ 2 \le i \le n-1.$ If both $ v_1 $ and $ w_1 $ are not in $ D $, then $ G[(D \setminus \{v_2\}) \cup \{v_1\}] $ and $ G[(D \setminus \{w_2\}) \cup \{w_1\}] $ are disconnected, which is a contradiction. Hence $ v_1 $ or $ w_1 $ is in $ D.$ Similarly, $ w_n $ or $ v_n $ is in $ D.$ We now claim that $ D \cap V_1 $ is a dominating set of $ P.$ Suppose there exists a vertex $ v_i $ such that $ v_i $ is not dominated by $ D \cap V_1.$ Then $ w_i \in D$ and $ G[(D \setminus\{w_i\}) \cup \{v_i\}] $ is disconnected, which is a contradiction. Hence $ D \cap V_1 $ is a dominating set of $ P.$ Similarly $ D \cap V_2 $ is a dominating set of $ Q.$ Now suppose $ D \cap V_1 \subsetneq V_1$ and $ D \cap V_2 \subsetneq V_2$. If three consecutive vertices of $ P $ say, $ v_i, v_{i+1}, v_{i+2} $ are not in $ D $, then $ w_i, w_{i+1}, w_{i+2} \in D.$ However, $ G[(D\setminus\{w_{i+1}\}) \cup \{v_{i+1}\}] $ is disconnected, which is a contradiction. Now suppose $ v_i, v_{i+1} \notin D$. Then $ v_{i-1}, v_{i+2}, w_{i},w_{i+1} \in D.$ Now since $ G[D]$ is connected, it follows that $ w_{i-1}, w_{i+2} \in D.$ Hence $ (D \setminus \{w_i,w_{i+1}\}) \cup \{v_i,v_{i+1}\} $ is also a SCDS of $ G.$ Thus by repeating the above process we get a SCDS of $ G,$ $ D_1 $ such that $ \vert D_1 \vert = \vert D \vert$, $ D_1 \cap V_1=V_1$ and $ D_1 \cap V_2 $ is a dominating set of $ Q.$ Thus, $\vert D \vert = \vert D_1 \vert \ge n+ \lceil \frac{n}{3} \rceil.$ Therefore, $ \gamma_{sc}(G)=n+\lceil\frac{n}{3} \rceil.$
\end{proof}
\noindent We now proceed to present results on algorithmic aspects such as NP-comppleteness and linear time algorithm for some classes of graphs.\\
\textit{Secure Connected Domination Problem} (\textit{SCDM}) \\
\indent \textbf{Instance}: A connected graph $G$ and a positive integer $l$.\\
\indent \textbf{Question}: Does there exist a SCDS of size at most $ l $ in $ G $ ? \bigskip  \\ 
\noindent The proof is by reduction from the Domination problem (DM), which is NP-complete \cite{garey}. \smallskip \\[3pt]
\textit{Domination Problem} (\textit{DM})\\
\indent \textbf{Instance}: A graph $G$ and a positive integer $k$.\\
\indent \textbf{Question}: Does there exist a DS of size at most $ k $ in $ G $ ?  
 
		\begin{thm}\label{scdmnp}
			SCDM is NP-complete.
		\end{thm}
		\begin{proof} 
			It can be easily verified that SCDM is in NP. 
			Now let $G=(V,E)$ be a graph and let $ k $ be a positive integer. Let $ G^* $ be the graph with $V(G^*)=V \cup \lbrace x \rbrace$ and $E(G^*) = E \cup \lbrace (u, x)$ $:$ $u \in V \rbrace $ and let $ l=k+1 $. Clearly, $G^*$ can be constructed from $G$ in polynomial time.   \par
			Now if $ D $ is a dominating set of $ G $ with $ \vert D \vert \le k $, then $ S=D \cup \{x\} $ is an SCDS of $ G^* $. Conversely, let $ S^* $ be an SCDS of $ G^* $ with $ \vert S^* \vert \le k+1.$ If $ x \in S $, then it follows from Proposition \ref{p1} that $ epn(x,S)=\emptyset $. Therefore, every vertex $u\in V(G^*)\setminus S$ is adjacent to a vertex in $S\setminus\{x\}$. Hence $S\setminus \{x\}$ is a DS of size at most $k$ in $G$. If $x \notin S$, Proposition \ref{p2}, the set $S \setminus \{v\}$, for any $v \in S$, is a DS of size at most $k$ in $G$. 
		\end{proof}
\noindent Next, we define the decision version of total domination and secure total domination problems as follows.\\ \vspace*{0.1cm} 
\textit{Total Domination Problem} (\textit{TDM}) \\
\indent \textbf{Instance}: A simple, undirected graph $G$ without isolated vertices and a positive integer $r$.\\
\indent \textbf{Question}: Does there exist a TDS of size at most $ r $ in $ G $ ? \\[4pt]
\textit{Secure Total Domination Problem} (\textit{STDM}) \\
\indent \textbf{Instance}: A simple, undirected and connected graph $G$ and a positive integer $m$.\\
\indent \textbf{Question}: Does there exist a STDS of size at most $ m $ in $ G $ ? 
\begin{thm}\label{stdmnp}
	STDM is NP-complete.
\end{thm}
\begin{proof}
	 It is clear that STDM is in NP. The reduction given in the proof of Theorem \ref{scdmnp} shows that STDM is NP-complete.  
\end{proof}

We now give NP-completeness results even when restricted to bipartite graphs or split graphs.
\noindent We formulate the SCDM for bipartite graphs as follows.\\[3pt]
\noindent \textit{Secure Connected Domination Problem for Bipartite Graphs} (\textit{SCDB})\\ [6pt]
\indent\textbf{Instance:} A connected bipartite graph $G = (V_1, V_2, E)$ and a positive integer $r$.\\
\indent\textbf{Question:} Does there exist a SCDS of size at most $ r $ in $ G $ ? 
\begin{thm}\label{scdbnp}
	SCDB is NP-complete.
\end{thm}
\begin{proof} It can be seen that SCDB is in NP. We transform an instance of SCDM problem to an instance of SCDB as follows. Given a graph $G$, we construct a graph $G^*(V_1, V_2, E)$ where $V_1(G^*)=V \cup  \lbrace p, q \rbrace$, $V_2(G^*)= V^\prime(G)$ $\cup$  $\lbrace x, y \rbrace$, here $V^\prime(G)$ is another copy of $V$ such that if $u$ and $v$ are two vertices in $V$ then the corresponding vertices in $V^\prime(G)$ are labeled as $u^\prime$ and $v^\prime$, and $E(G^*)$ consists of (i) edges $uv^\prime$ and $u^\prime v$ for each edge $uv \in E$; (ii) edges of the form $uu^\prime$ for each vertex $u \in V$; and (iii) edges of the form $ux$ and $uy$ for every vertex $u \in V_1(G^*)$. 	
	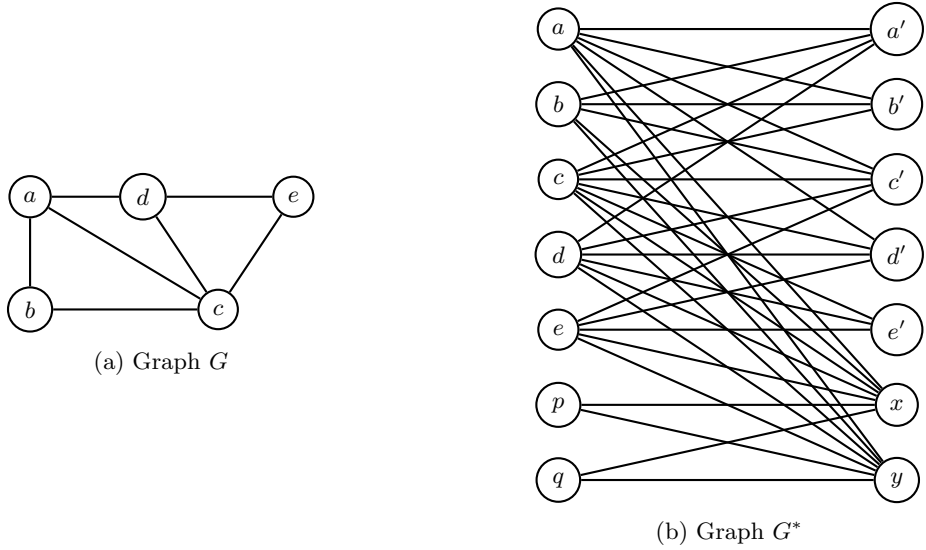
\begin{figure}[h]
		\centering
		\begin{subfigure}{.5\textwidth}
			\centering
				\begin{tikzpicture}[auto, node distance=1.5cm, every loop/.style={},
				thick,main node/.style={circle,draw,font=\sffamily\small\bfseries}]
				
				\node[main node] (1) {$a$};
				\node[main node] (2) [below of=1] {$b$};
				\node[main node,node distance=2.5cm] (3) [right of=2] {$c$};
				\node[main node] (4) [right of=1] {$d$};
				\node[main node,node distance=2.0cm] (5) [right of=4] {$e$};
								
				\path[every node/.style={font=\sffamily\small}]
				(1) edge node[left] {} (4)
				edge node[left] {} (2)
				edge node[left] {} (3)
			
				(2) 
								
				(3)  edge  node[left] {} (2)
				edge node[left] {} (4)
				edge node[left] {} (5)
				
				(4) 
				edge node[left] {} (5)

				;
				
				\end{tikzpicture}
				\caption{Graph $G$}

			\label{fig:sub1}
		\end{subfigure}%
		\hfill
		\begin{subfigure}{.5\textwidth}
			\centering
			\begin{tikzpicture}[auto, node distance=1.5cm, every loop/.style={},
			thick,main node/.style={circle,draw,font=\sffamily\small\bfseries}]
			
			\node[main node,node distance=1cm] (1) {$a$};
			\node[main node,node distance=1cm] (2) [below of=1] {$b$};
			\node[main node,node distance=1cm] (3) [below of=2] {$c$};
			\node[main node,node distance=1cm] (4) [below of=3] {$d$};
			\node[main node,node distance=1cm] (5) [below of=4] {$e$};
			\node[main node,node distance=1cm] (6) [below of=5] {$p$};
			\node[main node,node distance=1cm] (7) [below of=6] {$q$};
			
			\node[main node,node distance=4.5cm] (11) [right of=1]{$a^\prime$};
			\node[main node,node distance=1cm] (22) [below of=11] {$b^\prime$};
			\node[main node,node distance=1cm] (33) [below of=22] {$c^\prime$};
			\node[main node,node distance=1cm] (44) [below of=33] {$d^\prime$};
			\node[main node,node distance=1cm] (55) [below of=44] {$e^\prime$};
			
			\node[main node,node distance=1cm] (x) [below of=55] {$x$};
			\node[main node,node distance=1cm] (y) [below of=x] {$y$};
			
			\path[every node/.style={font=\sffamily\small}]
			(1) edge node[left] {} (11)
			edge node[left] {} (22)
			edge node[left] {} (33)
			edge node[left] {} (44)
			edge node[left] {} (x)
			edge node[left] {} (y)
			
			(2)	edge node[left] {} (11)
			edge node[left] {} (22)
			edge node[left] {} (33)
			edge node[left] {} (x)
			edge node[left] {} (y)
			
			(3)  edge  node[left] {} (11)
			edge node[left] {} (22)
			edge node[left] {} (33)
			edge node[left] {} (44)
			edge node[left] {} (55)
			edge node[left] {} (x)
			edge node[left] {} (y)
			
			(4) edge node[left] {} (11)
			edge node[left] {} (33)
			edge node[left] {} (44)
			edge node[left] {} (55)
			edge node[left] {} (x)
			edge node[left] {} (y)

			(5)edge node[left] {} (55)
			edge node[left] {} (44)
			edge node[left] {} (33)
			edge node[left] {} (x)
			edge node[left] {} (y)
			
			(6)edge node[left] {} (x)
			edge node[left] {} (y)

			(7)edge node[left] {} (x)
			edge node[left] {} (y)
			;
			\end{tikzpicture}
			\caption{Graph $G^*$}\label{fig:sub2}
		\end{subfigure}
		\caption{Construction of $G^*$ from $G$}
		\label{fig:test}
	\end{figure}
	Clearly $G^*$ is a bipartite graph and can be constructed from $G$ in polynomial time. \par
	Next, we show that $G$ has a SCDS of size at most $r$ if and only if $G^*$ has a SCDS of size at most $r+2$. If $ S $ is a SCDS of $ G $ with $ \vert S \vert \le r,$ then it can be easily verified that $ S^*= S \cup \{x,y\} $ is a SCDS of $ G^* $ with $ \vert S^* \vert \le r+2.$ \par
	Conversely, let $S^*$ be an SCDS of $G^*$ and  $\vert S^* \vert \le r+2$.  Since $ x $ and $ y $ are the only vertices in $ S^* $ which defend $ p $ and $ q $, it follows that at least one of them must be in $ S^*$. If $ x \in S^* $ and $ y \notin S^* $, then $ G^*[(S^*\setminus \{x\}) \cup \{p\}] $ is disconnected, which is a contradiction. Hence $ x,y \in S^*$. Let $ S^\prime= S^* \setminus \{x,y,p,q\} $ and $ S^{\prime\prime}=(S^\prime \cup \{v : v^\prime \in S^\prime \cap V^\prime(G)\}) \setminus \{v^\prime : v^\prime \in S^\prime \cap V^\prime(G)\}$. Clearly $ S^{\prime\prime} $ forms a SCDS of size at most $ r $ in $ G $. 
\end{proof}

\begin{thm}\label{stdmnp1}
	STDM is NP-complete for bipartite graphs.
\end{thm}
\begin{proof}
	 It is clear that STDM for bipartite graphs is in NP. The reduction given in the proof of Theorem \ref{scdbnp} shows that STDM is NP-complete for bipartite graphs.  
\end{proof}
\noindent Since the Domination problem is w[2]-complete for bipartite graphs \cite{vraman} and the reductions in Theorem \ref{scdbnp} and Theorem  \ref{stdmnp1} are in the function of the parameter $ l,$ the following two corollaries are immediate.
\begin{cor}
	SCDM is w[2]-hard in bipartite graphs.
\end{cor} 
\begin{cor}
	STDM is w[2]-hard in bipartite graphs.
\end{cor} 

\noindent It has been shown that the DM and the TDM as NP-complete even when restricted to split graphs\cite{bersto}. 
\begin{thm} \label{scdomsplit}
SCDM is NP-complete for split graphs.
\end{thm}
\begin{proof}
It is known that SCDM is a member of NP. We reduce DM for split graphs to SCDM for split graphs. Given a split graph $G$ whose vertex set is partitioned into a clique $Q$ and an independent set $I$, we construct a split graph $G^*$ with a clique $Q^*$ and an independent set $I^*$ as follows:\\
\indent $V(G^*) = V \cup \{ x ,y \}$, and \\
\indent $E(G^*) = E \cup \{ xu : u \in V \} \cup \{xy\}$.\\
\noindent Note that $G^*$ is a split graph, where $Q^* = Q \cup \{ x \}$ and $I^* = I \cup \{ y \}$ and $G^*$ can be constructed from $G$ in polynomial time. \par 
Now let $ S $ be a DS of $ G $ with $\vert S \vert \le k $. Then $ S^*=S \cup \{x,y\} $ is a SCDS of $ G^* $ with $ \vert S^* \vert \le k+2.$ Conversely, let$ S^* $ be a SCDS of $ G^* $ with $ \vert S^*\vert \le k+2.$ It follows from Proposition \ref{pendant} that $ x,y \in S^*.$ Clearly $S^\prime= S^* \setminus \{x,y\} $ is a DS of $ G $ with $ \vert S^\prime \vert \le k $. 
\end{proof}

\begin{thm}\label{stdomsplit}
STDM is NP-complete for split graphs.
\end{thm}
\begin{proof}
	 It is clear that STDM for split graphs is in NP. The reduction given in the proof of Theorem \ref{scdomsplit} shows that STDM is NP-complete for split graphs.  
\end{proof}
\noindent Since the Domination problem is w[2]-complete for split graphs \cite{vraman} and the reductions in Theorem \ref{scdomsplit} and Theorem \ref{stdomsplit} are in the function of the parameter $ c,$ the following two corollaries are immediate.
\begin{cor}
	SCDM is w[2]-hard in split graphs.
\end{cor} 
\begin{cor}
	STDM is w[2]-hard in split graphs.
\end{cor} 
\noindent In the next two theorems we prove that $ \gamma_{sc}(G) $ can be computed in linear time for block graphs and threshold graphs.\par
Let $ G=(V,E) $ be a connected graph. A vertex $ v $ is called a cut-vertex of $ G $ if $ G - v $ is a disconnected graph. A graph $ G $ with no cut-vertex is called a block. A block $ B $ of a graph is a maximal connected induced subgraph of $ G $ such that $ B $ has no cut-vertex. In block $ B $, vertices which are not cut vertices of $ G $ are called block vertices. A graph $ G $ is called a \textit{block graph} if all its blocks are cliques. 

\begin{definition}
	A graph $ G $ is called a block graph if all the blocks of $ G $ are cliques.
\end{definition}
\begin{thm}\label{path}
	Let $ G $ be a block graph having $ r $ blocks and $ k $ cut vertices. Then $ \gamma_{sc}(G) = k+r-r^\prime$, where $ r^\prime$ is the number of blocks such that all vertices of the block are cut vertices.
\end{thm}
\begin{proof}
	Let $ A $ denote the set of all cut vertices of $ G$. Let $ B_1, B_2, \ldots,B_{r^\prime},B_{r^\prime+1}, \ldots, B_r $ be the blocks of $ G $ where every vertex of $ B_i $ is a cut vertex of $ G $ if $ 1 \le i \le r^\prime.$ Let $ T=\{v_i : 1 \le i \le r-r^\prime$ and $ v_i$ is a non-cut vertex of $ B_{r^\prime+i}\}$. Let $ S=A \cup T.$ Since $ A$ contains all cut-vertices of $ G $, it follows that $ G[S] $ is connected. Also if $ v \in V \setminus S $, then $ v $ is not a cut-vertex. Now there exists a vertex $ u \in T $ such that $ uv \in E $ and $ (S \setminus \{v\}) \cup \{u\}$ is a CDS of $ G.$ Hence, $ S $ is a SCDS of $ G.$ Therefore, $ \gamma_{sc}(G)\le k+r-r^\prime.$   \par
	Now let $ D $ be any $\gamma_{sc} $-set of $ G.$ Since $ G[D] $ is connected, $ D \supseteq A.$ Further, a cut-vertex cannot defend any other vertex and hence $ D $ contains at least one non-cut vertex from each block $ B_i $ where $ r^\prime+1 \le i \le r.$ Hence $\gamma_{sc}(G)= \vert D \vert \ge \vert S \vert=k+r-r^\prime.$ Thus $\gamma_{sc}(G)= k+r-r^\prime.$
\end{proof}
\begin{cor}
	Let $ G $ be a block graph with $ r $ blocks and exactly one cut-vertex. Then $ \gamma_{sc}(G)=r+1.$ \\
\end{cor}
\begin{cor}
	For any tree $ T$ with $ n $ vertices, $ \gamma_{sc}(T)=n.$
\end{cor}
\begin{proof}
	Here $ r=n-1, r^\prime=n-1-l$ and $ k=n-l$ where $ l$ is the number of leaves in $ T.$ Therefore, $ \gamma_{sc}(T)=n.$
\end{proof}
\begin{cor}
	SCDM is linear time solvable for block graphs.
\end{cor}
\begin{proof}
	Since number of blocks and number of cut-vertices of block graph can be determined in linear time, the result follows.
\end{proof}
\begin{definition}
	A graph $ G=(V, E) $ is called a \textit{threshold graph} if there is a real number $ t $ and a real number $ w(v) $ for every $ v \in V $ such that a set $ S \subseteq V $ is independent if and only if $ \sum_{v \in S}w(S) \le t $. 
\end{definition}
	Threshold graphs considered are assumed to be non-complete and connected. We use the following characterization of threshold graphs given in \cite{threshold1} to prove that secure connected domination number can be computed in linear time for threshold graphs.\\
\indent  A graph $ G $ is a \textit{threshold graph} if and only if it is a split graph and for split partition $(C, I) $ of $ V $, there is an ordering $ (x_1, x_2, \ldots, x_p) $ of vertices of $ C $ such that $ N[x_1] \subseteq N[x_2] \subseteq \ldots \subseteq N[x_p],$ and there is an ordering $ (y_1, y_2, \ldots, y_q) $ of the vertices of $ I $ such that $ N(y_1)\supseteq N(y_2) \supseteq \ldots N(y_q). $
\begin{thm}
	Let $ G $ be a connected threshold graph. Then $ \gamma_{sc}(G)=2+l$, where $ l $ is the number of pendant vertices.
\end{thm}
\begin{proof}
	Let $ S=\{x_p,x_{p-1}\} \cup \{v\in I : v\in N(x_p) \setminus N(x_{p-1}) \}$. Clearly $ G[S]$ is a star with center $ x_p.$ Also every vertex $ v \in V \setminus S $ is defended by $ x_{p-1} $ and $ G[(S \setminus \{x_{p-1}\}) \cup \{v\}] $ is connected. Thus, $ S $ is a SCDS of $ G.$ Hence $ \gamma_{sc}(G) \le 2+l.$ \par
	Let $ D $ be any $ \gamma_{sc}$-set of $ G$. It follows from Theorem \ref{pendant} that $  \vert D \vert \ge l+1.$ If $ \vert D \vert =l+1, $ then exactly one vertex of $ C $ say, $ u $ is a support vertex. Hence no vertex of $ C \setminus \{u\} $ is $ D $-defended, which is a contradiction. Hence $\gamma_{sc}(G)=\vert D \vert \ge 2+l.$ Thus $\gamma_{sc}(G)= 2+l.$
\end{proof}
\begin{thm}
	SCDM is linear time solvable for threshold graphs.
\end{thm}
\begin{proof}
		Since the ordering of the vertices of the clique in a threshold graph can be determined in linear time \cite{threshold1}, the result follows.
\end{proof}
\section{Conclusion}
\indent In this paper, it is shown that secure connected (total) domination problem is NP-complete even when restricted to bipartite graphs, or split graphs. Since split graphs form a proper subclass of chordal graphs, these problems are also NP-complete for chordal graphs. We have proved that secure connected domination problem is linear time solvable for block graphs and threshold graphs. It will be interesting to investigate the algorithmic complexity of secure connected (total) domination problem for subclasses of chordal and bipartite  graphs. 

\end{document}